\documentclass[a4paper, abstract, 11pt]{scrartcl}
\usepackage[UKenglish]{babel}
\usepackage[rgb]{xcolor}
\usepackage{subcaption}
\usepackage{combelow}
\usepackage{tikz}
\usetikzlibrary{fit}
\usepackage{setspace}
\usepackage{csquotes}
\usepackage{microtype}
\usepackage{amsmath, amsthm,thmtools,thm-restate}
\usepackage{a4wide}
\usepackage[libertine, sb]{newtx}
\usepackage{inconsolata}
\usepackage{bm}
\usepackage{enumitem}
\usepackage[roman]{complexity}
\usepackage[colorlinks=true,allcolors=linkColor, hypertexnames=false]{hyperref} 
\usepackage{nameref}
\usepackage[nameinlink]{cleveref}

\renewcommand{\LIN}{\ensuremath{\textup{LIN}}}

\newcommand\bigmid{\ensuremath{\mathrel{\Bigg|}}}
\renewcommand\epsilon\varepsilon
\renewcommand\emptyset\varnothing

\newcommand{\X}{\mathcal{X}}
\newcommand{\Y}{\mathcal{Y}}
\renewcommand{\M}{\mathcal{M}}
\newcommand{\F}{\mathbb{F}}

\DeclareMathOperator\myC{\mathscr{C}}
\DeclareMathOperator\CSP{\mathrm{CSP}}
\DeclareMathOperator\ar{\mathrm{ar}}

\DeclareMathOperator\codim{codim}

\definecolor{color1}{Hsb}{100, 0.8, 0.8}
\definecolor{color2}{Hsb}{220, 0.8, 0.8}
\definecolor{color3}{Hsb}{340 , 0.8, 0.8}
\definecolor{linkColor}{Hsb}{100, 0.9, 0.40}

\declaretheorem[shaded={bgcolor=linkColor!5}]{theorem}
\declaretheorem[sibling=theorem, shaded={bgcolor=linkColor!5}]{lemma, corollary, proposition, conjecture}

\theoremstyle{definition}
\declaretheorem[sibling=theorem, shaded={bgcolor=linkColor!5}]{definition}

\theoremstyle{remark}
\declaretheorem[sibling=theorem, shaded={bgcolor=linkColor!5}]{remark}

\newenvironment{innerproof}[1]{
\begin{proof}[Proof of #1]\renewcommand{\qedsymbol}{\ensuremath{\textit{(End of proof of #1)\quad $\blacksquare$}}}}{\end{proof}}

\usepackage{todonotes}

\title{Strong Sparsification for 1-in-3-SAT \\  via Polynomial Freiman-Ruzsa\thanks{An extended abstract of part of this work appeared in Proceedings of FOCS 2025. This work was supported by UKRI EP/X024431/1 and a Clarendon Fund Scholarship. Work done while Benjamin Bedert and Tamio-Vesa Nakajima were at the University of Oxford. For the purpose of Open Access, the authors have applied a CC BY public copyright licence to any Author Accepted Manuscript version arising from this submission. All data is provided in full in the results section of this paper.}}

\author{Benjamin Bedert\\
University of Cambridge\\
\url{bb741@cam.ac.uk}
\and
Tamio-Vesa Nakajima\\
University of Marburg\\
\url{nakajima@uni-marburg.de}
\and
Karolina Okrasa\\
University of Oxford\\
\url{karolina.okrasa@cs.ox.ac.uk}
\and
Stanislav \v{Z}ivn\'{y}\\
University of Oxford\\
\url{standa.zivny@cs.ox.ac.uk}}

\setkomafont{author}{\large}

\begin{document}
\maketitle
\begin{abstract}
We introduce a new notion of sparsification, called \emph{strong sparsification}, in which constraints are not removed but variables can be merged. As our main result, we present a strong sparsification algorithm for 1-in-3-SAT. The correctness of the algorithm relies on establishing a sub-quadratic bound on the size of certain sets of vectors in $\F_2^d$. This result, obtained using the recent \emph{Polynomial Freiman-Ruzsa Theorem} (Gowers, Green, Manners and Tao, Ann. Math. 2025), could be of independent interest. As an application, we improve the state-of-the-art algorithm for approximating linearly-ordered colourings of 3-uniform hypergraphs (H{\aa}stad, Martinsson, Nakajima and {\v{Z}}ivn{\'{y}}, APPROX 2024).
We also investigate the existence of strong sparsification algorithms for other constraint satisfaction problems.
\end{abstract}

\section{Introduction}

\emph{Sparsification}, the idea of reducing the size of an object of interest (such as a graph or a formula) while preserving its inherent properties, has been tremendously successful in many corners of computer science. 

One notion of sparsification comes from the influential paper of Bencz\'ur and
Karger~\cite{Benczur96:stoc}, who showed that, for any $n$-vertex graph $G$, one
can efficiently find a weighted subgraph $G'$ of $G$ with $O(n\log n)$ many
edges, so that the size of \emph{all} cuts in $G$ is approximately preserved in
$G'$, up to a small multiplicative error. The bound on the size of $G'$ was
later improved to linear by Batson, Spielman and
Srivastava~\cite{Batson14:siam}. The optimality of the dependency on
$\varepsilon$ was established by the works of Andoni, Chen, Krauthgamer, Qin,
Woodruff and Zhang~\cite{Andoni16:itcs} and Carlson, Kolla, Srivastava and Trevisan
\cite{DBLP:conf/soda/CarlsonKST19}. From the many
follow-up works, we mention the paper of Kogan and
Krauthgamer~\cite{Kogan15:itcs}, who initiated the study of sparsification for
constraint satisfaction problems (CSPs), Filtser and Krauthgamer~\cite{Filtser17:sidma}, who classified sparsifiable Boolean binary CSPs, and Butti and \v{Z}ivn\'y~\cite{Butti20:sidma}, who classified sparsifiable binary CSPs on all finite domains.
Some impressive progress on this line of work has been made in recent years by Khanna, Putterman and Sudan~\cite{Khanna24:soda,Khanna24:focs,Khanna25:stoc}, establishing optimal sparsifiers for classes of CSPs and codes, and Brakensiek and Guruswami~\cite{Brakensiek25:stoc}, who pinned down the sparsifiability of all CSPs (up to polylogarithmic factors, and non-efficiently). A different but related notion of sparsification is the concept of (unweighted) additive cut sparsification, introduced by Bansal, Svensson and Trevisan~\cite{Bansal19:focs},
later studied for other CSPs by Pelleg and \v{Z}ivn\'y~\cite{Pelleg24:talg}.

Another study on sparsification comes from computational complexity. The \emph{Exponential Time Hypothesis} (ETH) of Impagliazzo, Paturi and Zane~\cite{Impagliazzo01:jcss} postulates that 3-SAT requires exponential time: there exists $\delta>0$ such that an $n$-variable 3-SAT instance requires time $O(2^{\delta n})$.\footnote{The weaker hypothesis $\P \neq \NP$ only postulates that 3-SAT requires super-polynomial time.} The \emph{sparsification lemma} from the same paper~\cite{Impagliazzo01:jcss} is then used to establish that ETH implies that exponential time is needed for a host of other problems. The lemma roughly says that any $n$-variable 3-SAT instance is equisatisfiable to an OR of exponentially many 3-SAT formulae, each of which has only linearly many clauses in $n$.\footnote{The precise statement includes a universal quantification over an arbitrarily small $\epsilon>0$ that controls the growth of the exponentials involved.} This should be contrasted with what can (or rather cannot) be done in polynomial time: under the assumption that $\NP\not\subseteq \coNP/\poly$, Dell and Melkebeek showed that 3-SAT cannot be sparsified in polynomial time into an equivalent formula with $O(n^{3-\epsilon})$ clauses~\cite{Dell14:jacm}. 

Drawing on techniques from fixed-parameter tractability~\cite{Cygan15:book} and
kernelisation~\cite{Fomin2019kernelization,Cygan15:book}, Jansen and
Pieterse~\cite{Jansen19:toct} and Chen, Jansen and
Pieterse~\cite{Chen20:algorithmica} studied which \NP-complete Boolean CSPs
admit a non-trivial sparsification. In particular, they observed that 1-in-3-SAT
admits a linear-size sparsifier, meaning an equivalent instance with $O(n)$ many
clauses, where $n$ is the number of variables~\cite{Jansen19:toct}. Building on
techniques from the algebraic approach to CSPs, Lagerkvist and Wahlstr\"om then
considered CSPs over domains of larger size~\cite{Lagerkvist20:toct}. These
results were recently generalised by Brakensiek, Guruswami, Jansen, Lagerkvist
and Wahlstr{\"{o}}m~\cite{DBLP:journals/corr/abs-2507-07942} in their study of
the so-called \emph{non-redundant} CSP instances.

In the present article we will be interested in sparsifying \emph{approximate} problems. When dealing with \NP-hard problems, there are two natural ways to relax the goal of exact solvability and turn to approximation: a quantitative one and a qualitative one. The first one seeks to maximise the number of satisfied constraints. A canonical example is the max-cut problem: finding a cut of maximum size is \NP-hard, but a cut of size at least roughly 0.878 times the optimum can be efficiently found by the celebrated result of Goemans and Williamson~\cite{Goemans95:jacm}. The second goal seeks to satisfy all constraints but in a weaker form. Here are a few examples of such problems. Firstly, the approximate graph colouring problem, studied by Garey and Johnson in the 1970s~\cite{GJ76}: given a $k$-colourable graph $G$, find an $\ell$-colouring of $G$ for some $\ell \geq k$. Secondly, finding a satisfying assignment to a $k$-SAT instance that is promised to admit an assignment satisfying at least $\lceil k/2\rceil$ literals in each clause --- a problem coined  $(2+\epsilon)$-SAT by Austrin, Guruswami and H{\aa}stad~\cite{AGH17}. Finally, given a satisfiable instance of (monotone) 1-in-3-SAT, find a satisfying not-all-equal assignment~\cite{BG21}. The former, quantitative notion of approximation has led to many breakthroughs in the last three decades, including the Probabilistically Checkable Proof (PCP) theorem~\cite{Arora98:jacm-probabilistic,Arora98:jacm-proof,Dinur07:jacm}. The latter, qualitative notion has been investigated systematically only very recently under the name of \emph{Promise Constraint Satisfaction Problems} (PCSPs)~\cite{AGH17,BG21,BBKO21}.
Unfortunately, traditional notions of sparsification, involving removing constraints, fail when applied to qualitative approximation.\footnote{Although recent work on the so-called \emph{conditional non-redundancy}~\cite{Brakensiek25:stoc,DBLP:journals/corr/abs-2507-07942} may offer some new techniques here.}

\begin{figure}[t]
    \begin{subfigure}[t]{.48\textwidth}
    \centering
    \begin{tikzpicture}[scale=.95]
        \begin{scope}[local bounding box=G1]
            \node[circle, draw] (A) at (0, 0) {};
            \node[circle, draw] (B) at (1, 0) {};
            \node[circle, draw] (C) at (0, 1) {};
            \node[circle, draw] (D) at (1, 1) {};
            \draw(A) -- (B) -- (D) -- (C) -- (A);

            \node at (0.5, 0.5) {$G$};
        \end{scope}
        
        \begin{scope}[shift={(2.5, 0)}, local bounding box=G2]
            \node[circle, draw] (A) at (0, 0) {};
            \node[circle, draw] (B) at (1, 0) {};
            \node[circle, draw] (C) at (0, 1) {};
            \node[circle, draw] (D) at (1, 1) {};
            \draw(C) -- (A) (B) -- (D) -- (C);
            \draw[dashed] (A) -- (B);

            \node at (0.5, 0.5) {$G'$};
        \end{scope}
        
        \begin{scope}[shift={(2.5, -2.5)}, local bounding box=G3]
            \node[circle, draw, fill=color1!85] (A) at (0, 0) {};
            \node[circle, draw, fill=color1!85] (B) at (1, 0) {};
            \node[circle, draw, fill=color2!85] (C) at (0, 1) {};
            \node[circle, draw, fill=color3!85] (D) at (1, 1) {};
            \draw(C) -- (A) (B) -- (D) -- (C);
            \draw[dashed] (A) -- (B);
        \end{scope}
        
        \begin{scope}[shift={(0, -2.5)}, local bounding box=G4]
            \node[circle, draw, fill=color1!85] (A) at (0, 0) {};
            \node[circle, draw, fill=color1!85] (B) at (1, 0) {};
            \node[circle, draw, fill=color2!85] (C) at (0, 1) {};
            \node[circle, draw, fill=color3!85] (D) at (1, 1) {};
            \draw(C) -- (A) (B) -- (D) -- (C);
            \draw[color3, very thick] (A) -- (B);
        \end{scope}

        \draw[gray, thick, ->, shorten <= .3cm, shorten >= .3cm] (G1) -- (G2);
        \draw[gray, thick, ->, shorten <= .3cm, shorten >= .3cm] (G2) -- (G3);
        \draw[gray, thick, ->, shorten <= .3cm, shorten >= .3cm] (G3) -- (G4);
    \end{tikzpicture}
    \caption{}\label{fig:1}
    \end{subfigure}%
    \begin{subfigure}[t]{.48\textwidth}
    \centering
    \begin{tikzpicture}[scale=0.95]
        \begin{scope}[local bounding box=G1]
            \node[circle, draw] (A) at (0, 0) {};
            \node[circle, draw] (C) at (0, 1) {};
            
            \node[circle, draw] (B) at (1, 0) {};
            \node[circle, draw] (D) at (1, 1) {};
            \draw(A) -- (B) -- (D) -- (C) -- (A);
            \node[fill=gray, rotate fit=45, fit=(A) (D), rounded corners, fill opacity=0.3] {};
            \node[fill=gray, rotate fit=45, fit=(B) (C), rounded corners, fill opacity=0.3] {};
            
            \node at (0.5, 0.5) {$G$};
        \end{scope}
        
        \begin{scope}[shift={(2.5, 0.2)}, local bounding box=G2]
            \node[circle, draw] (A) at (0, 0) {};
            \node[circle, draw] (B) at (1, 0) {};
            \draw (A) -- (B);
            \node at (0.5, 0.5) {$G'$};
        \end{scope}
        
        \begin{scope}[shift={(2.5, -2)}, local bounding box=G3]
            \node[circle, draw, fill=color1!85] (A) at (0, 0) {};
            \node[circle, draw, fill=color2!85] (B) at (1, 0) {};
            \draw (A) -- (B);
        \end{scope}
        
        \begin{scope}[shift={(0, -2.5)}, local bounding box=G4]
            \node[circle, draw, fill=color1] (A) at (0, 0) {};
            \node[circle, draw, fill=color2] (C) at (0, 1) {};
            
            \node[circle, draw, fill=color2] (B) at (1, 0) {};
            \node[circle, draw, fill=color1] (D) at (1, 1) {};
            \draw(A) -- (B) -- (D) -- (C) -- (A);
            \node[fill=color1!85, rotate fit=45, fit=(A) (D), rounded corners, fill opacity=0.3] {};
            \node[fill=color2!85, rotate fit=45, fit=(B) (C), rounded corners, fill opacity=0.3] {};
        \end{scope}

        \draw[gray, thick, ->, shorten <= .3cm, shorten >= .3cm] (G1) -- (G2);
        \draw[gray, thick, ->, shorten <= .3cm, shorten >= .3cm] (G2) -- (G3);
        \draw[gray, thick, ->, shorten <= .3cm, shorten >= .3cm] (G3) -- (G4);
    \end{tikzpicture}
    \caption{}\label{fig:2}
    \end{subfigure}
    \caption{}
\end{figure}

To illustrate that, consider the following naive procedure for graph 2-colouring: given an $n$-vertex instance graph~$G$, if $G$ \emph{is} bipartite then return a spanning forest of $G$; if $G$ is \emph{not} bipartite then return one of the odd cycles in $G$. This simple algorithm is essentially the best possible: it outputs an instance $G'$ with at most $n$ edges, whose set of 2-colourings is exactly the same as that of $G$. However, the above approach breaks down for approximate solutions: there are 3-colourings of $G'$ that are not 3-colourings of $G$ (see \Cref{fig:1}). 

Therefore, we need a notion that allows us to turn approximate solutions of our simplified instance into the solutions of the original one. In the above example of 2-colouring, a desired outcome would be a sparse graph $G'$ that is 2-colourable if and only if $G$ is and, furthermore, any $k$-colouring of $G'$ translates into a $k$-colouring of $G$. Luckily, it is easy to see how to do this here. Suppose there exist two vertices with a common neighbour in $G$, say $x - y - z$. Then, $x$ and $z$ must be coloured identically in \emph{all} 2-colourings of $G$. Hence, we can identify $x$ with $z$; that is, we replace $x$ and $z$ with a new vertex $x'$, both in the vertex set of $G$ and the edges of~$G$. Let $G'$ be the result of applying the identification procedure iteratively for all such triples. Now, if $G$ was originally 2-colourable, so is $G'$. (Indeed, there is a 1-to-1 correspondence between the 2-colourings of $G$ and those of $G'$.) Furthermore, any $k$-colouring of $G'$ can be easily extended to a $k$-colouring of $G$: we colour each vertex of $G$ according to the colour of the vertex it was merged into in $G'$ (see \Cref{fig:2}). 

Observe that the key property of the procedure outlined above is that, since all we do is merge variables that are equal in all solutions, no constraints are deleted. There is nothing special about 2-colourings in this argument --- an analogous method, which we call a \emph{strong sparsification}, can be applied to other computational problems, including variants of graph and hypergraph colouring problems, cf.~\Cref{sec:general}.
We remark that this idea was used, under the name ``progress type 3'', by
Blum~\cite{Blum94:jacm} in his seminal paper on on approximate graph colouring.

Here, we focus on a particular generalisation of 2-colouring, namely (monotone) 1-in-3-SAT: given a set of variables $X=\{x_1, \ldots, x_n\}$, together with a set $C \subseteq X^3$ of clauses, assign values $0$ and $1$ to the variables so that for every clause $(x_i, x_j, x_k) \in C$ exactly one variable among $x_i, x_j, x_k$ is set to $1$, with the remaining two set to $0$.\footnote{A strong sparsification algorithm for monotone 1-in-3-SAT can be generically transformed into one for \emph{non-monotone} 1-in-3-SAT (i.e.~allowing negated literals), cf.~\Cref{sec:reduction}. Thus, we shall focus on monotone 1-in-3-SAT.}

\begin{definition}\label{def:ss}
    A \emph{strong sparsification algorithm for monotone 1-in-3-SAT} is an algorithm which, given an instance~$\X = (X,C)$ of monotone 1-in-3-SAT, outputs an equivalence relation $\sim$ on $X$ such that if $x_i \sim x_j$ then $x_i$ and $x_j$ have the same value in all solutions to $\X$. 
    We define the instance $\X / \mathord{\sim} = (X / \mathord{\sim},C/\mathord{\sim})$ of monotone 1-in-3-SAT as follows: $X/ \mathord{\sim}$ is the set of the equivalence classes $[x_1]_\sim, \ldots, [x_n]_\sim$, and each clause $(x_i, x_j, x_k) \in C$ induces a clause $([x_i]_\sim, [x_j]_\sim, [x_k]_\sim)$ in $C / \mathord{\sim}$.
    The \emph{performance} of the algorithm is given by the number of clauses in $\X / \mathord{\sim}$, as a function of $n=|X|$.
\end{definition}
We emphasise that the notion of strong sparsification can be defined in an analogous way for other satisfiability problems. However, since a strong sparsification is, in particular, a sparsification in the sense of the aforementioned work of Dell and van Melkebeek~\cite{Dell14:jacm}, for some classic problems of this type (3-SAT in particular), it is unlikely to obtain any non-trivial results. We focus on 1-in-3-SAT, one of the first problems for which positive sparsification results were obtained~\cite{Jansen19:toct}.

The trivial strong sparsification for monotone 1-in-3-SAT (the one that does not merge any variables) has worst-case performance $O(n^3)$. There is a slightly more clever approach that has performance $O(n^2)$ (noted e.g.~in~\cite{DBLP:journals/tcs/DroriP02}). Suppose there exist clauses $(x, y, z)$ and $(x, y, t)$. There are only 3 possible assignments to $(x, y, z, t)$: $(1, 0, 0, 0)$, $(0, 1, 0, 0)$ and $(0, 0, 1, 1)$. We see immediately that $z$ and $t$ are the same in all solutions and thus can be merged. After the exhaustive application of this rule, we get an instance in which, for every pair of variables $x, y$, there is at most one $z$ such that $(x,y,z) \in C$, so the number of clauses is $O(n^2)$ (if we interpret monotone 1-in-3-SAT instances as 3-uniform hypergraphs, such hypergraphs are called \emph{linear}). 

With this approach, the $O(n^2)$ bound is essentially tight: let $X=\{0,1\}^d$, and $C=\{(i, j, k) \mid i, j, k \in X, i \oplus j \oplus k = 0 \}$, where $\oplus$ denote the bitwise Boolean XOR operation. It is not difficult to see that, for any $i, j \in X$, there is exactly one variable $k$ (namely $k={i \oplus j}$) such that $(i,j,k)$ is a clause. Hence, the above strong sparsification does nothing and outputs the original instance with $\Theta(n^2)$ clauses. 

It is much harder to find a strong sparsification with better than quadratic performance, and in fact the existence of such an algorithm is not a priori clear. 
As our main contribution, we show that such an algorithm exists and thus improve the trivial quadratic upper bound.

\begin{restatable}[\textbf{Main}]{theorem}{thmSparsAlg} \label{thm:SparsAlg}
    There exists a polynomial-time strong sparsification algorithm for monotone 1-in-3-SAT with performance $O(n^{2 - \epsilon})$, for $\epsilon \approx 0.0028$.
\end{restatable}

A proof of \Cref{thm:SparsAlg} can be found in \Cref{sec:algorithm}. The main technical ingredient is the following theorem, which could be of independent interest. It is proved in \Cref{sec:addComb} using tools from additive combinatorics.

\begin{restatable}{theorem}{thmAddComb}\label{thm:AddComb}
    Fix $n, d$. Consider $V = \{ v_1, \ldots, v_n \}\subseteq \F_2^d$ and let $N_1, \ldots, N_n \subseteq V$ satisfy
    \begin{enumerate}[label=(\roman*)]
        \item for all $i\in[n]$, $v_i+N_i=N_i$, and
        \item for all $i\in[n]$, $v_1,v_2\dots,v_{i-1}\notin \langle N_i+N_i\rangle$.
    \end{enumerate}
    Then $\sum_{i = 1}^n |N_i| = O(n^{2 - \epsilon})$ for $\epsilon \approx 0.0028$.
\end{restatable}

Note that the answer is polynomial in $n$ since, for example, a linear lower bound is achieved by taking $N_i=\{0,v_i\}$.
In fact, in the remark following \Cref{prop:model}, we provide an example where $\sum_i|N_i|=\Omega(n^{\log_23})$.
Note also that, since $|N_i|\leq n$, there is a trivial bound
$\sum_i|N_i|\leq n^2$. Hence, our contribution consists in improving the trivial bound by a polynomial saving of $n^\epsilon$. 

We also show, in~\Cref{sec:general}, that no algorithm (even one with exponential
runtime) can output a strong sparsifier with performance $o(n^{1.725\ldots})$.
This is because there are instances of monotone 1-in-3-SAT with $n$ variables
and $\Omega(n^{1.725\ldots})$ constraints in which no merges of any two variables are
possible.

\paragraph{Application} 

As an application of our result, we improve the state-of the-art approximation of hypergraph colourings.
There are several different notions of colourings for hypergraphs, the classic
one being nonmonochromatic colourings~\cite{DRS05}. We shall focus on
\emph{linearly-ordered} (LO) colourings~\cite{Barto21:stacs}, also known as \emph{unique-maximum} colourings~\cite{Cheilaris13:sidma}: the colours are taken from a linearly-ordered set, such as the integers, with the requirement that the maximum colour in each hyperedge is unique. 

Notice that finding an LO 2-colouring of a 3-uniform hypergraph $H$ is precisely the same problem as monotone 1-in-3-SAT (interpret the clauses of an instance as edges of $H$, and take the order $0 < 1$). Hence, our strong sparsification algorithm from \Cref{thm:SparsAlg} also applies to LO 2-colouring 3-uniform hypergraphs.
Thus, we can improve the state-of-the-art algorithms for
\emph{approximate LO colouring}, where we are given an LO
2-colourable 3-uniform hypergraph $H$, and are asked to find an LO-colouring with as few colours as possible. The best known algorithm so far is the following, due to H{\aa}stad, Martinsson, Nakajima and {\v{Z}}ivn{\'{y}}.

\begin{theorem}[\protect{\cite[Theorem~1]{HMNZ24_logarithmic}}]
There is a polynomial-time algorithm that, given an $n$-vertex 3-uniform LO
  2-colourable hypergraph $H$ with $n \geq 4$,\footnotemark{} returns an LO $(\log_2
  n)$-colouring of $H$.
\end{theorem}
\footnotetext{Assumptions like this are to make sure that $\log_2 n \geq 2$. We will have similar assumptions for other algorithms.}

Using our sparsification algorithm, we improve this to the following. 

\begin{restatable}{corollary}{coralg}\label{cor:alg}
There is a polynomial-time algorithm that, given an $n$-vertex 3-uniform LO 2-colourable hypergraph $H$ with $n \geq 5$, returns an LO $(0.999 \log_2 n)$-colouring of $H$. 
\end{restatable}

We will use the following result from~\cite{HMNZ24_logarithmic}. Since its performance depends on the number of edges in the input, it benefits from our sparsification scheme.

\begin{theorem}[\protect{\cite[Theorem~3]{HMNZ24_logarithmic}}]\label{thm:oldAlg}
There is a polynomial-time algorithm that, given a 3-uniform LO 2-colourable hypergraph $H$ with $m$ edges, returns an LO $(2 + \frac{1}{2}\log_2 m)$-colouring of $H$.
\end{theorem}

\begin{proof}[Proof of \Cref{cor:alg}]
    Recall that a hypergraph $H = (V, E)$ can be interpreted as an instance $\X=(V,E)$ of monotone 1-in-3-SAT, where $V$ is the set of variables, the clauses of $\X$ are the edges of $H$, and any solution to $\X$ correspond to an LO 2-colouring of $H$ and vice-versa.
   Thus, using \Cref{thm:SparsAlg}, we compute an equivalence relation $\sim$ on $V$ so that, if $u \sim v$, then $u$ and $v$ get the same colour in any LO 2-colouring; and $H' = H /\mathord{\sim}$ has $O(n^{2 - \epsilon})$ edges for $\epsilon \approx 0.0028$. Since $H$ is LO 2-colourable, $H'$ is as well (and, in fact, the LO 2-colourings of $H$ and $H'$ are in a 1-to-1 correspondence). 
    
    Next, using \Cref{thm:oldAlg}, find an LO $(2 + \frac{1}{2} \log_2( O(n^{2 - \epsilon})))$-colouring of $H'$, i.e.~an LO colouring with $O(1) + \frac{2 - \epsilon}{2} \log_2 n$ colours. Note that $(2 - \epsilon) / 2 \approx 0.9986$, and so, for $n$ larger than some constant, we have $O(1) + \frac{2 - \epsilon}{2} \log_2 n \leq 0.999 \log_2 n$; whereas, for $n$ smaller than a constant, we can find an LO 2-colouring by brute force. 
    Since for every $n \geq 5$ any LO 2-colouring is in particular an LO $(0.999 \log_2 n)$-colouring, in all cases we have found an LO $(0.999 \log_2 n)$-colouring of $H'$. Note that any LO colouring of $H' = H / \mathord{\sim}$ immediately gives rise to an LO colouring of $H$, by assigning each vertex of $H$ the colour given to its equivalence class in $H / \mathord{\sim}$. 
\end{proof}

Finally, we remark that the introduced notion of strong sparsification might be
worth exploring for other satisfiability problems and CSPs. 
This could unveil a new and
exciting line of work, going beyond the results of the present article. 
We start this exploration in~\Cref{sec:general}, where we obtain bounds for some
CSPs, including monotone 1-in-$k$-SAT and, more generally, $\ell$-in-$k$-SAT,
Not-All-Equal-$k$-SAT, graph $k$-colouring, and systems of linear equations.

\paragraph{Paper structure}
\Cref{sec:addComb} gives a proof of the main technical result, \Cref{thm:AddComb}. 
\Cref{sec:algorithm} gives a proof of our sparsification algorithm, \Cref{thm:SparsAlg}.
\Cref{sec:reduction} shows that monotone strong sparsification implies non-monotone strong sparsification.
Finally, \Cref{sec:general} shows generalisations of our work to other CSPs and studies the mere existence of strong sparsifiers, independently of the allowed time to find them.

\paragraph{Acknowledgements} We thank the anonymous reviewers of FOCS 2025 for
their very useful and detailed feedback on an extended abstract of this work. We thank Alexandru Pascadi for introducing the authors to each other. We also thank Rare\cb{s}-Darius Buhai (aided by chatGPT) for pointing out the paper~\cite{Kane17:ejc}.

\section{Additive combinatorics}\label{sec:addComb}
In this section we prove \Cref{thm:AddComb}. We reformulate it into an equivalent, notationally more convenient statement as follows.
\begin{theorem}\label{th:polysaving}
     Let
    \[
     \myC(n)\ \coloneq\ \max_{V,N_i}\left\{\sum_{i=1}^n|N_i| \bigmid V\subseteq\F_2^d, |V|=n, \text{ and $N_1,\dots,N_n\subseteq V$ satisfy (i), (ii) of \Cref{thm:AddComb}\ }\right\}.
    \]
    Then, $\myC(n)=O(n^{2-\epsilon})$ for $\epsilon=0.0028$. 
\end{theorem}
(We note in passing that in fact we show
$\myC(n) \leq 2^{548} n^{2 - \frac{1}{74}\log_4(4/3)}$.)
 
Our proof of this bound employs two prominent results from the area of additive combinatorics, namely the Balog-Szemer\'edi-Gowers and Polynomial Freiman-Ruzsa Theorems. In particular, the value of $\epsilon$ that we obtain with this approach depends (essentially linearly) on the strongest known constants in these theorems. 
To state them, we need to introduce some standard concepts from additive combinatorics.

For a set $A\subseteq \F_2^d$ and an integer $k$, we define its \emph{$k$-energy} 
\[
    E_k(A)\coloneq \#\{(a_1,a_2,\dots,a_k)\in A^k\mid a_1+a_2+\dots+a_k=0\}.
\]
Note that if $r_A(x)\coloneq\#\{(a_1,a_2)\in A^2\mid a_1+a_2=x\}$, then we can equivalently write $E_4(A)=\sum_{x\in\F_2^d}r_A(x)^2$.\footnote{We remark that $E_4(A)=\#\{(a_1,a_2,a_3,a_4)\in A^4\mid a_1+a_2=a_3+a_4\}$ is commonly known as the \emph{additive energy} of $V$ in the additive combinatorics community.} It may be helpful to keep in mind the trivial upper bound $|E_k(A)|\leq |A|^{k-1}$ which holds since after choosing $a_1,\dots,a_{k-1}$, the element $a_k$, if exists, is fixed by the equation. We also define the \emph{sumset} $A+A\coloneq\{a_1+a_2 \mid a_j\in A\}$, and the ratio $|A+A|/|A|$, known as the \emph{doubling constant} of $A$. Again, there are the trivial bounds $|A|\leq |A+A|\leq |A|^2$. One should think of sets with ``large energy'' (say $E_4(A)\geq |A|^3/K$) and sets with ``small doubling'' (say $|A+A|\leq K|A|$) as being highly additively structured. The latter notion is strictly stronger, and it is not hard to show that a set $B$ with doubling constant $L$ automatically satisfies $E_4(B)\geq |B|^3/L$. 

For a set $A \subseteq \F^d_2$ we denote by $\langle A \rangle$ the smallest subspace that contains all elements of $A$.
If $A=\{a_1,\ldots,a_q\}$, we omit the internal brackets and write $\langle a_1,\ldots,a_q \rangle$ instead of $\langle \{a_1,\ldots,a_q\}\rangle$.

The following Balog-Szemer\'edi-Gowers Theorem is a standard result in additive combinatorics~\cite{Balog94:comb,Gowers98:gafa}, and provides a partial converse. Roughly speaking, it states that a set with large
  additive energy contains a rather large subset with small doubling. We will use a recent version due to Reiher and Schoen with the current best dependence on $K$.
\begin{theorem}[Balog-Szemer\'edi-Gowers
  Theorem~\cite{Reiher24:comb}]
    Let $K\geq 1$  and $\epsilon \in (0, \frac{1}{2})$ and let $B \subseteq \F^d_2$ have additive energy $E_4(B)\geq |B|^3/K$. Then there exists a subset $B'\subseteq B$ of size $|B'|\geq (1 - \epsilon) |B|/K^{1/2}$ with doubling $|B'+B'|\leq 2^{33} \epsilon^{-9} K^4 |B'|$.
\end{theorem}
By setting $\epsilon= 1/8$, we get the following streamlined corollary which is what we use.
\begin{corollary}\label{th:balog}
    Let $K\geq 1$ and let $B \subseteq \F^d_2$ have additive energy $E_4(B)\geq |B|^3/K$. Then there exists a subset $B'\subseteq B$ of size $|B'|\geq 7|B|/(8K^{1/2})$ with doubling $|B'+B'| \leq 2^{60} K^4 |B'|$.
\end{corollary}

Another very recently proved celebrated theorem, known as the Polynomial Freiman-Ruzsa Conjecture or Marton's Conjecture, describes the structure of sets $B$ with small doubling in $\F_2^d$, stating that they must essentially be contained in a small number of translates of a subgroup.
\begin{theorem}[Polynomial Freiman-Ruzsa Theorem~\cite{Gowers25:annals}]\label{th:freiman}
    Let $K\geq 1$ and let $B\subseteq\F_2^d$ be a set with doubling $|B+B|\leq K|B|$. Then there exists a subspace $G\leq \F_2^d$ such that $|G|\leq |B|$ and $B$ is contained in at most $2K^9$ translates of $G$.
\end{theorem}

We begin with proving~\Cref{th:polysaving} in the special case where $V=\F_2^d$ is a vector space itself, stated below as~\Cref{prop:model}. The argument in this setting relies on the so-called polynomial method, and does not require the two theorems above. The power of the Balog-Szemer\'edi-Gowers and Polynomial Freiman-Ruzsa Theorems will be required to deal with general sets $V$, essentially by reducing the general problem in \Cref{th:polysaving} to this special setting (or more accurately \Cref{cor:model}).
\begin{proposition}\label{prop:model}
    Let $V=\{v_1,v_2,\dots,v_n\}= \F_2^d$ and $N_1,\dots,N_n\subseteq V$ satisfy conditions (i) and (ii) from \Cref{thm:AddComb}.
    Then $\sum_{i=1}^n|N_i| \leq n^{\log_23}$, where we note that $\log_23\approx 1.585$.
\end{proposition}
\begin{remark}\label{rem:lowerbound}
Interestingly, the above bound is optimal as the following example shows. Let $V=\F_2^d$; by abuse of notation we let $S \subseteq [d]$ denote the indicator vector for the set $S \subseteq [d]$. Thus $V = \{ S \mid S \subseteq [d] \}$; also we have implicitly defined addition on sets to be the symmetric difference. We order $V$ in decreasing order of size i.e.~$S$ comes before $T$ if $|S| > |T|$. Next, define $N_{S} \coloneq \{ T \in V \mid T \subseteq S \}$. Note that (i) is satisfied: $S + N_{S} = N_{S}$, since if $T \subseteq S$ then $S + T = (S \setminus T) \subseteq S$.
One can also check that (ii) is satisfied because if $T \in N_{S}$ for $S \neq T$, then $S \supset T$, hence $|S| > |T|$ which implies that $T$ comes after $S$ in our ordering. Finally, $|N_{S}|=2^{|S|}$ (it contains one vector for each subset of $S$) and we calculate $\sum_{S \subseteq [d]} |N_{S}|= \sum_{i = 0}^n \binom{d}{i} 2^i = 3^d = (2^d)^{\log_2 3}$.
\end{remark}
Before giving the proof, we note that \Cref{prop:model} implies a power saving bound whenever $V$ is contained in subspace $H$ whose  size is not much larger than that of $V$ itself.

\begin{corollary}\label{cor:model}
Let $V =\{v_1,v_2,\dots,v_n\}\subseteq \F_2^d$ and $N_1, \ldots, N_n \subseteq V$ satisfy conditions (i) and (ii) from \Cref{thm:AddComb}. If $V\subseteq H$ is contained in a subspace $H$, then $\sum_{i=1}^n|N_i|\leq 2|H|^{\log_23}$.
\end{corollary}
\begin{proof}[Proof of \Cref{cor:model}]
    Define $V'=\{v_1,\dots,v_n,v_{n+1},\dots,v_{|H|}\}$ where $v_{n+1},\dots,v_{|H|}$ is an arbitrary ordering of the elements of $H\setminus V$. Let $N'_i=N_i$ if $i\in[n]$ and $N'_i=\emptyset$ if $i>n$. By \Cref{prop:model}, the claimed bound clearly follows. 
\end{proof}
\begin{proof}[Proof of \Cref{prop:model}]
    Let the sets $V=\F_2^d$ and $N_i$ be given. Since conditions (i) and (ii) (as well as the sizes $|N_i|$) are preserved under translating $N_i$ (i.e.~replacing $N_i$ by $N_i + x$ for some $x \in \F_2^d$), we may assume without loss of generality that $0 \in N_i$ for each $i$. Hence, $\langle N_i+N_i\rangle=\langle N_i\rangle \eqcolon H_i$ for subspaces $H_i \leq \F_2^d$. We may further assume without loss of generality that for each $i$ we have $N_i=H_i$. Indeed, condition (ii) remains unaffected, while $N_i+v_i=N_i$ implies that $\langle N_i\rangle +v_i=\langle N_i\rangle$ (and replacing $N_i$ by $\langle N_i\rangle$ can also only increase the sizes $|N_i|$).

    Thus, it is enough to show that, if $v_1,v_2,\dots,v_n$ is an ordering of $\F_2^d$ and $H_i$ are subspaces such that
    \begin{enumerate}[label={(\roman*)}]
        \item $v_i\in H_i$, and
        \item $v_1,v_2,\dots,v_{i-1}\notin H_i$,
    \end{enumerate}
    then $\sum_i|H_i| \leq 2 n^{\log_23}$. Let us write $h_i= \codim H_i = d-\dim H_i,$
    so that for each $i$ we may find $h_i$ many vectors $u^{(i)}_1,u^{(i)}_2,\dots, u^{(i)}_{h_i} \in \F_2^d$ for which
    \begin{equation}\label{eq:Hdefi}
        H_i=\{x\in\F_2^d \mid x \cdot u^{(i)}_r = 0 \text{ for all $r=1,2,\dots, h_i$}\}.
    \end{equation}
    We emphasise that for $x,y\in\F_2^d$ we write $x\cdot y=\sum_{j=1}^dx_jy_j \in \F_2$. 
    
    Consider the following polynomial in the variable $X=(X_1,X_2,\dots,X_d)$ for each $i\in[n]$:
    \begin{equation}\label{eq:F_jdefi}
        F_i(X)=F_i(X_1,X_2,\dots,X_d) \coloneq \prod_{r=1}^{h_i}\left(X\cdot u^{(i)}_r -1\right), 
    \end{equation}
    which is simply a product of $h_i$ many linear polynomials. Since we will only ever evaluate this polynomial for $X\in\F_2^d$, we may employ a \emph{multilinearisation} trick which replaces each occurrence of a power $X_s^t$ by $X_s$, for $t=1,2,3,\dots$ and $s\in[n]$. One should note that this does not affect evaluations of $F_i(X)$ for $X\in\F_2^d$, since $Y^t=Y$ for $Y\in\F_2$. Multilinearizing each $F_i(X)$, we obtain the polynomials $\Tilde{F}_i$ which are linear combinations of monomials in $d$ variables, having degree $0$ or $1$ in each variable:
    \[
        \Tilde{F}_i(X)\in \langle 1,X_1,\dots,X_d, X_1X_2,\dots,X_1X_2\dots X_d\rangle.
    \]
    By construction, $F_i(X)=\Tilde{F}_i(X)$ for each $X\in\F_2^d$. The crucial property of these polynomials is the following.
    \begin{lemma}\label{lem:polyeval}
        We have that 
        $\displaystyle\tilde{F}_j(v_i)=F_j(v_i) =
        \begin{cases}
            0 &\text{ if $i<j$},\\
            1 &\text{ if $i=j$}.
        \end{cases}$
    \end{lemma}
    
    \begin{innerproof}{\Cref{lem:polyeval}}
        Note that if $i<j$, then by condition (ii) we have $v_i\notin H_j$ and hence from \eqref{eq:Hdefi} there exists $r\in\{1,2,\dots, h_j\}$ such that $v_i\cdot u^{(j)}_r=1$ . Clearly \eqref{eq:F_jdefi} then shows that $F_j(v_i)=0$. Now, if $i=j$, then by condition (i) we have $v_i=v_j\in H_i$ which by \eqref{eq:Hdefi} means precisely that $v_i\cdot u^{(i)}_r=0$ for all $r\in\{1,2,\dots, h_i\}$. Hence, $F_j(v_i)=1$.
    \end{innerproof}
    \Cref{lem:polyeval} easily implies the following.
    \begin{lemma}\label{lem:polys}
        The polynomials $\tilde{F}_1,\tilde{F}_2,\dots,\tilde{F}_n$ are linearly independent in the polynomial vector space given by $\langle 1, X_1,\dots,X_k, X_1X_2,\dots, X_1X_2\dots X_d\rangle$.
    \end{lemma}
    \begin{innerproof}{\Cref{lem:polys}}
        Suppose not. Then there is a dependence relation
        \[
            \tilde{F}_{j_1}+\tilde{F}_{j_2}+\dots+\tilde{F}_{j_m}=0
        \]
        for some $1\leq j_1<\dots<j_m\leq n$. But then, evaluating this polynomial at $X=v_{j_1}\in\F_2^d$ would give a contradiction by \Cref{lem:polyeval}: $0= \tilde{F}_{j_1}(v_{j_1})+\tilde{F}_{j_2}(v_{j_1})+\dots+\tilde{F}_{j_m}(v_{j_1})=1+0+\dots+0 =1$.
    \end{innerproof}
    To use this information to bound $\sum_{i=1}^n|N_i|=\sum_i|H_i|$, we find a bound on the sizes of the \emph{level sets} $V_{\alpha} \coloneq \{i\in[n]\mid |H_i|\geq \alpha n=\alpha 2^d\}$ for each $\alpha\in [0,1]$. As each $H_i$ is a subspace, it suffices to bound $|V_\alpha|$ when $\alpha=2^{-b}$ for some $b\in\{0,\ldots,d\}$. Now note that if $i\in V_{2^{-b}}$, then $|H_i|\geq 2^{d-b}$ so that $h_i=\codim H_i\leq b$. Hence, from the definition \eqref{eq:F_jdefi} we see that the collection of polynomials
    \[
        \{\tilde{F}_i(X)\mid i\in V_{2^{-b}}\}\subseteq \left\langle \prod_{i=1}^dX_i^{m_i} \bigmid m_i\in\{0,1\} \text{ and }\sum_{i=1}^dm_i\leq b\right\rangle \eqcolon P_{b}
    \]
    is a set of $|V_{2^{-b}}|$ many linearly independent polynomials which are all contained in the subspace $P_{b}\leq \langle 1,X_1,\dots,X_1X_2,\dots,X_1\dots X_d\rangle$ of polynomials of total degree at most $b$. This implies
    \[
        |V_{2^{-b}}|\leq \dim P_{b} =\sum_{r=0}^{b}\binom{d}{r}.
    \]

Since sets $V_{2^{-b}}$ are nested, to avoid counting their elements more than once, we define $U_{2^{0}}=V_{2^0}$, and $U_{2^{-b}} \coloneq V_{2^{-b}} \setminus U_{2^{-(b-1)}}$ for every $b \in [d]$. 
By definition, the sets $U_{2^{-b}}$ now contains precisely these $i \in [n]$ for which $|H_i|=2^{d-b}$ and therefore $|U_2^{-b}|\leq {d \choose b}$. 
Hence, in total we can bound
\begin{multline*}
\sum_{i=1}^n|N_i|=\sum_{i=1}^n|H_i|=
\sum_{b=0}^d\sum_{i \in U_{2^{-b}}} |H_i|
\leq \sum_{b=0}^d 2^{d-b}|U_{2^{-b}}|
\leq\sum_{b=0}^{d}2^{d-b}\binom{d}{b}=3^d=(2^d)^{\log_23}=n^{\log_23}.
\end{multline*}
%Hence, in total we can bound
%\begin{multline*}
%\sum_{i=1}^n|N_i|=\sum_{i=1}^n|H_i|\leq \sum_{b=0}^d2^{d-b}|V_{2^{-b}}|
%\leq\sum_{b=0}^{d}2^{d-b}\sum_{r=0}^{b}\binom{d}{r}
%=\sum_{r=0}^{d}\binom{d}{r}(1+2+\dots+2^{d-r})
%\leq 2\sum_{r=0}^{d}\binom{d}{r}n/2^r
%\\
%= 2n (3/2)^{d}
%\end{multline*}
%Since $2^d = n$, we note that this final sum is identical to $2 \cdot 3^d = 2 n^{\log_23}$, as required.
which is precisely the desired bound.
\end{proof}

We proceed to the proof of the general case.

\begin{proof}[Proof of \Cref{th:polysaving}]
We will prove the bound $\myC(n) \leq 2^{548} n^{2 - \epsilon}$ for all $n \in \mathbb{N}$. In particular, 
we will proceed by induction on $n$, assuming that the bound $\myC(m)\leq 2^{548} m^{2-\epsilon}$ holds for all $m<n$.

Suppose that the set $V=\{v_1,v_2,\dots,v_n\}\subseteq\F_2^d$ and sets $N_1,\dots,N_n\subseteq V$ satisfy conditions (i), (ii) and are such that $\sum_{i=1}^n|N_i|= \myC(n)$. We may assume that \begin{equation}\label{eq:largeC}\sum_{i=1}^n|N_i|\geq n^{2-\epsilon},\end{equation} as otherwise we are done. The first step in this proof consists in showing, using tools from additive combinatorics, that under the assumption \eqref{eq:largeC}, a large subset of $V$ must be rather densely contained in a subspace of $\F_2^d$. We show that \eqref{eq:largeC} implies that $E_3(V)$, and hence $E_4(V)$, are large.

\begin{lemma}\label{lem:largeenergy}
    Let $V$ and the sets $N_i\subseteq V$ satisfy \eqref{eq:largeC}. Then $E_3(V)\geq \sum_{i=1}^n|N_i|\geq n^{2-\epsilon}$. Moreover, $E_4(V)\geq n^{3-2\epsilon}$
\end{lemma}
\begin{innerproof}{\Cref{lem:largeenergy}}
    The bound for $E_3(V)$ is trivial from condition (i), since whenever $j,k\in [n]$ are such that $v_k\in N_j$, then $v_j+v_k\in N_j\subseteq V$ so that $(v_j,v_k,v_j+v_k)$ is a tuple that contributes to $E_3(V)$.

    Note that $\sum_{v\in V}r_V(v)=E_3(V)\geq n^{2-\epsilon}$. Also, we observed above that $\sum_{x}r_V(x)^2=E_4(V)$. By Cauchy-Schwarz, we then get
  \[ E_4(V)=\sum_{x\in\F_2^d}r_V(x)^2\geq\sum_{v\in V}r_V(v)^2
  \geq \frac{1}{|V|}\left(\sum_{v\in V}r_V(v)\right)^2
  =\frac{E_3(V)^2}{n}\geq
  n^{3-2\epsilon}. \qedhere \]
\end{innerproof}

We may now combine \Cref{lem:largeenergy} with \Cref{th:balog} and \Cref{th:freiman}. By \Cref{th:balog} (Balog-Szemer\'edi-Gowers) and as $E_4(V)\geq n^{3-2\epsilon}=n^3/K$ for $K=n^{2\epsilon}$, there exists a subset $A\subseteq V$ of size $|A|\geq 7n^{1-\epsilon}/8$ with $|A-A|=2^{60} n^{8\epsilon}|A|$. Now, by \Cref{th:freiman} (Polynomial-Freiman-Ruzsa), we may find a subspace $H\leq \F_2^d$ such that $A$ is covered by $2^{541} n^{72\epsilon}$ translates of $H$, and where $|H|\leq |A|$. In particular, there is one such translate $x_0+H$ such that
\[
    |V\cap(x_0+H)|\geq \frac{|A|}{2^{541} n^{72\epsilon}}\geq \frac{n^{1-73\epsilon}}{2^{542}}.
\]
Also, without loss of generality we may take $x_0=0$, as otherwise we may replace $H$ by $\langle H \cup \{x_0\}\rangle$, which still satisfies the two properties above up to an additional factor of $2$, namely:
\begin{itemize}
    \item $|H|\leq 2|A|$,
    \item 
    $\displaystyle|V \cap H|\geq \frac{|A|}{2^{541} n^{72\epsilon}}\geq \frac{n^{1-73\epsilon}}{2^{542}}$.
\end{itemize}
Thus we have completed the first step of the proof. To use this information for estimating $\sum_i|N_i|$, we split $N_i=N_i^H\cup N_i^C$ for each $i\in[n]$, where
\[
    N_i^H=N_i\cap (H\cup (H+v_i))=N_i\cap\langle H \cup \{v_i\}\rangle
\]
and $N_i^C=N_i\setminus N_i^H$. It is notationally convenient to also define $N_v=N_i$ if $v=v_i\in V$, and similarly for $N_v^H,N_v^C$. We can calculate
\begin{equation}\label{eq:inductivebound}
\sum_{i=1}^n|N_i|= \sum_{v\in V\cap H} |N_v^H|+\sum_{v\in V\cap H}|N_v^C|+\sum_{v\in V\setminus H}|N_v^H|+\sum_{v\in V\setminus H}|N_v^C|.
\end{equation}
The reason for the definitions of the sets $N_i^H,N_i^C$ will become clear shortly: essentially, the idea is that, because $V\cap H$ is rather dense in the subspace $H$ by the first step, one may expect to obtain good bounds for the first three terms by applying \Cref{cor:model}. The final term may be bounded using the induction hypothesis, since it will be clear from our choice that the sets $N_v^C$ still satisfy conditions (i), (ii). The second step therefore consists in making this approach precise and bounding each of the four terms above.

\begin{enumerate}
    \item First, we immediately deduce from \Cref{cor:model} that
    \[
        \sum_{v\in V\cap H}|N_v^H|\leq 2|H|^{\log_23},
    \]
    since the sets $V'\coloneq V\cap H$ (ordered in the same way as in $V$) and $N'_v \coloneq N_v^H$ for $v\in V\cap H$ still satisfy conditions (i) and (ii). Only that $N_v^H+v=N_v^H$ for $v\in V\cap H$ is perhaps non-trivial, but this is satisfied since $N_{v}+v=N_{v}$ holds for the original sets $N_v$ and, as $v\in H$, we may take the intersection of both sides with $H$.

\item The final term may be bounded by $\sum_{v\in V\setminus H}|N_v^C|\leq \myC(|V\setminus H|)$, since the set $\tilde{V}\coloneq V\setminus H$ with $\tilde{N}_v\coloneq N_v^C\subseteq V\setminus H$ for $v\in \tilde{V}$ is again a system satisfying conditions (i), (ii). Indeed, (ii) is straightforward as $\tilde{N}_v\subseteq N_v$. Moreover, (i) holds: if $v\in V\setminus H$, then, as $v+N_v=N_v$, the set $N_v$ consists of pairs $x,x+v$. Recall that, by definition, $N_v^C=N_v\setminus (N_v\cap \langle H \cup \{v\}\rangle)$, thus we have indeed also only removed elements in pairs (i.e.~$x\in N_v\cap \langle H \cup \{v\}\rangle$ if and only if $x+v \in N_v\cap \langle H \cup \{v\}\rangle$). Therefore, as we showed above that $|V\cap H|\geq 2^{-542} n^{1-73\epsilon}$, and as $\myC(n)$ is clearly increasing in $n$, we can bound
\[
    \sum_{v\in V\setminus H}|N_v^C|\leq \myC(n-2^{-542}n^{1-73\epsilon}).
\]
    
\item To bound the middle sums in \eqref{eq:inductivebound}, we will use the following lemma, whose proof we postpone to the end of the section. 
\begin{lemma}\label{lem:middlesums}
    We have that
    \begin{enumerate}[label={(\alph*)}]
    \item $\displaystyle\sum_{v\in V\setminus H}|N_v^H| \leq 12 n|H|^{\frac{1}{2}\log_23}$, 
    \item $\displaystyle\sum_{v\in V\cap H}|N_v^C|= 12 n|H|^{\frac{1}{2}\log_23}$.
    \end{enumerate}
\end{lemma}

\end{enumerate}
It remains to show how the three bounds above may be combined to complete the proof of \Cref{th:polysaving} (and hence \Cref{thm:AddComb}). Using these bounds in \eqref{eq:inductivebound}, we get
\[
    \sum_{i=1}^n|N_i| \leq 2|H|^{\log_23}+24n|H|^{\frac{1}{2}\log_23}+\myC(n-2^{-542} n^{1-73\epsilon}).
\]
Recall that $|H|\leq 2|A|\leq 2n$ and that we assumed that $\myC(n)=\sum_{i=1}^n|N_i|$. In particular $2|H|^{\log_23} \leq 2 (2n)^{\log_23} \leq 6 n^{1 + \frac{1}{2} \log_23}$and $24 n |H|^{\frac{1}{2} \log_23} \leq 24 n (2n)^{\frac{1}{2} \log_23} \leq 42 n^{1 + \frac{1}{2} \log_23}$, since $24 \cdot 2^{\frac{1}{2} \log_23} = 24\sqrt{3} \leq 42$.
Thus, we conclude that 
\[
    \myC(n)
    \leq \myC(n(1- 2^{-542} n^{-73\epsilon})) + 48 n^{1+\frac{1}{2}\log_23}
    \leq \myC(n(1- 2^{-542} n^{-73\epsilon})) + 2^6 n^{1+\frac{1}{2}\log_23}.
\]
    By the induction hypothesis, we may bound 
\[\myC(n(1-2^{-542} n^{-73\epsilon}))\leq 2^{548} n^{2-\epsilon}(1-2^{-542} n^{-73\epsilon})^{2-\epsilon}
\leq
2^{548} n^{2-\epsilon}(1-2^{-542} n^{-73\epsilon})
\leq 2^{548}n^{2-\epsilon}-2^{6} n^{2-74\epsilon},\] 
since certainly $(1-2^{-542} n^{-73\epsilon})^{2-\epsilon}\leq 1-2^{-542} n^{-73\epsilon}$. We deduce that 
\[
    \myC(n)\leq 2^{548} n^{2 - \epsilon} - 2^6 n^{2 - 74 \epsilon} + 2^6 n^{1 + \frac{1}{2} \log_2 3}.
\]
This implies the desired bound $\myC(n)\leq 2^{548} n^{2- \epsilon}$ so long as $n^{2 - 74 \epsilon} \geq n^{1 + \frac{1}{2} \log_2 3}$.
If we choose $\epsilon>0$ such that $2-74\epsilon \geq 1+\frac{1}{2}\log_23$, then this is clearly true for all $n \in \mathbb{N}$. Hence, taking $\epsilon$ with this property suffices --- 
and any value of $\epsilon$ not exceeding $\frac{1-\frac{1}{2}\log_23}{74} = \frac{1}{74} \log_4 (4/3) \approx 0.002804$ works. This concludes the proof of \Cref{th:polysaving}.
\end{proof}

Our final task is then to prove \Cref{lem:middlesums}.

\begin{proof}[Proof of \Cref{lem:middlesums}]
\begin{enumerate}[wide, label={(\alph*)}]
    \item Recall that we want to bound $\sum_{v\in V\setminus H}|N_v^H|$ where $N_v^H=N_v\cap\langle H \cup \{v\}\rangle$. We begin with an estimate for the contributions of the sets $(V\setminus H)_a \coloneq \{v\in V\setminus H\mid |N_v^H|\geq a\}$ for each $a\in[n]$. 
    If $v\in V\setminus H$, then $N_v^H\subseteq N_v\subseteq V$, and, as $N_v^H=N_v\cap (H\cup (v+H))$ satisfies $N_v^H=v+N_v^H$, by condition (i), it follows that exactly half of the elements of $N_v^H=v+N_v^H$ lie in $H$ and the other half lie in $v+H$. 
    This means that if $v\in (V\setminus H)_a$, then the non-trivial coset $v+H$ intersects $V$ in at least $|N_v^H|/2\geq a/2$ elements. 
    
    Let $y_1+H,y_2+H,\dots,y_\ell+H$ be all the distinct non-trivial cosets of $H$ which each contain at least $a/2$ elements of $V$. We note two things. First, observe that $(V\setminus H)_a\subseteq \bigcup_{j=1}^\ell (y_j+H)$. Indeed, if $v\in(V\setminus H)_a$, then $v\in v+H$, and we proved above that $v+H$ contains at least $a/2$ elements of $V$.
    Second, since distinct cosets are disjoint, it is clear that $\ell \leq 2n/a$ as $|V|\leq n$. 
    We claim that if we fix one such large coset, call it $y+H$, then $\sum_{v\in y+H}|N_v^H| \leq 6|H|^{\log_23}$. By the two observations above this gives the bound $\sum_{v\in (V\setminus H)_a}|N_v^H|\leq\sum_{j=1}^\ell\sum_{v\in y_j+H}|N_v^H| \leq 6|H|^{\log_23}\ell \leq 12|H|^{\log_23}n/a$.

    To see why the claim holds, simply note that we may apply \Cref{cor:model} to the set $V' \coloneq V\cap( H\cup (y+H))=V\cap \langle H\cup \{y\}\rangle$, with the sets $N'_v \coloneq N_v^H\subseteq V'$ for all $v\in V'$. It is easy as always to see that these satisfy conditions (i), (ii), and note also that $|\langle H\cup \{y\}\rangle|=2|H|$, so that \Cref{cor:model} gives $\sum_{v\in y+H}|N_v^H|=2 |\langle H\cup \{y\}\rangle|^{\log_23} = 2 \cdot 2^{\log_23} \cdot |H|^{\log_23} = 6 |H|^{\log_23}$.

    From the definition of $(V\setminus H)_a$ we also know that $\sum_{v\in V\setminus (V\setminus H)_a}|N_v^H|\leq a|V\setminus (V\setminus H)_a|\leq na$. Finally, we can bound in total
    \begin{align*}
        \sum_{v\in V\setminus H}|N_v^H|&\leq\sum_{v\in V\setminus (V\setminus H)_a}|N_v^H|+\sum_{v\in (V\setminus H)_a}|N_v^H|\\
        &\leq na + 12|H|^{\log_23}n/a\\
        &\leq 12 |H|^{\frac{1}{2}\log_23}n, 
    \end{align*}
    if we choose $a=(|H|^{\log_23})^{1/2}$.
    \item Next, we need to find a good way to bound $\sum_{v \in V\cap H}|N_v^C|$, where we recall that $N_v^C=N_v\setminus (H\cup (v+H))$. Note in particular that $N_v^C\subseteq V\setminus H$. We again proceed by considering estimates for the contributions of the level sets 
    \[
        (V\setminus H)^{(a)} \coloneq \{w\in V\setminus H \mid w\text{ appears in at least $a$ many sets $N_v^C$ with $v\in V\cap H$}\}.
    \]
    These are perhaps slightly more complicated than the level sets above; note that these level sets are not subsets of the set $V\cap H$ over which we are summing, but of $V\setminus H$. However, it is clear that for any $a$: \begin{align*}
        \sum_{v\in V\cap H}|N_v^C|&=\sum_{v\in V\cap H}|N_v^C\setminus(V\setminus H)^{(a)}|+\sum_{v\in V\cap H}|N_v^C\cap(V\setminus H)^{(a)}|\\
        &\leq na+\sum_{v\in V\cap H}|N_v^C\cap(V\setminus H)^{(a)}|,
    \end{align*} so it is sufficient to find, for each $a$, good bounds on $\sum_{v\in V\cap H}|N_v^C\cap(V\setminus H)^{(a)}|$. 
    
    Similarly as above, if we fix $a\in[n]$ then we claim that we may find cosets $y_1+H,\dots,y_\ell+H$ such that $(V\setminus H)^{(a)}\subseteq \bigcup_{j=1}^\ell (y_j+H)$ and $\ell \leq 2n/a$. Indeed, it is enough to take all the distinct cosets $y_j+H$ which each contain at least $a/2$ elements of $V$, and note first that there clearly are at most $2n/a$ such cosets as in the proof of the first bound. To see why $(V\setminus H)^{(a)}\subseteq \bigcup_{j=1}^\ell(y_j+H)$, pick a $w\in (V\setminus H)^{(a)}$ and recall that by definition, $w\in N_v^C\subseteq N_v$ for at least $a$ many $v\in V\cap H$. By condition (i), this means that $w+v\in N_v\subseteq V$ for at least $a$ many vectors $v\in H$, so that the coset $w+H$ contains at least $a$ elements of $V$. This completes the proof of the claim.
    
    Again, we apply \Cref{cor:model} for each subspace $\langle H \cup \{y_i\}\rangle$ with the set $V' \coloneq V\cap\langle H \cup 
\{y_i\}\rangle$ (ordered in the same way as in $V$) and $N'_v \coloneq N_v\cap\langle H \cup \{y_i\}\rangle\subseteq V'$ for $v\in V\cap H$ (and to be fully rigorous we may take $N'_v=\emptyset$ for $v\in V\cap(y_i+H)$). 
The condition (ii) is clearly satisfied, and since for every $v \in V \cap H$ we have $v+(N_v \cap (H \cup (H+y_i))) = (N_v \cap (H \cup (H+y_i))$ as $v+N_v =N_v$, (i) also holds.
    We deduce that $\sum_{v\in V\cap H}|N_v^C\cap(y_i+H)|\leq \sum_{v\in V'}|N'_v|\leq2|\langle H\cup \{y_i\}\rangle|^{\log_23} \leq 2 \cdot 2^{\log_23} |H|^{\log_23} = 6 |H|^{\log_23}$, since $N_v^C\cap (y_i+H)\subseteq N'_v$. Hence, summing over all $\ell\leq 2n/a$ cosets, we get \[\sum_{v\in V\cap H}|N_v^C\cap(V\setminus H)^{(a)}|\leq \sum_{j=1}^\ell\sum_{v\in V\cap H}|N_v^C\cap(y_j+H)|= 12 |H|^{\log_23}n/a,\] 
    and, in total,
    \begin{align*}
        \sum_{v\in V\cap H}|N_v^C|&\leq an+\sum_{v\in V\cap H}|N_v^C\cap(V\setminus H)^{(a)}|\\
        &\leq an+12|H|^{\log_23}n/a \leq 12 |H|^{\frac{1}{2}\log_23}n,
    \end{align*}
    if we choose $a=(|H|^{\log_23})^{1/2}$. That concludes the proof of \Cref{lem:middlesums}.\qedhere
\end{enumerate}
\end{proof}

We remark that the polynomial method was also used in the aforementioned recent work of Brakensiek and Guruswami on CSP sparsification~\cite[Section 6]{Brakensiek25:stoc}.
\section{Sparsification algorithm}\label{sec:algorithm}

In this section, we prove \Cref{thm:SparsAlg}, i.e.~present an efficient algorithm for strong sparsification of monotone 1-in-3-SAT.
For convenience, we consider the problem of monotone \emph{2-in-3}-SAT which is obtained by swapping the roles of 0 and 1 in the definition of 1-in-3-SAT --- it is clear that for our purposes these two problems are equivalent. 
Exploiting the ideas of \cite{HMNZ24_logarithmic, NZ24_LO}, we will make use of the \emph{linear structure} of 2-in-3-SAT: a clause $(x, y, z)$ of a monotone 2-in-3-SAT instance is satisfied if and only if $x + y + z = 2$, $x, y, z \in \{0, 1\}$.

\begin{definition}
    Consider an instance $\X = (X,C)$ of monotone 2-in-3-SAT. We define a system of modulo~2 linear equations $A_\X$ as follows. The set of variables of $A_\X$ is $X$, and for every clause $(x, y, z) \in C$, $A_\X$ contains the linear equation $x + y + z \equiv 0 \bmod 2$.
\end{definition}

Clearly $A_\X$ is a relaxation of $\X$ --- every solution to $\X$ is also a solution to $A_\X$; in particular, if two variables are always equal in every solution to $A_\X$, then they are always equal in every solution to $\X$. We say that two distinct variables $x$ and $y$ are \emph{twins}
if $\hat{x} = \hat{y}$ for every solution $(\hat{v})_{v \in X}$ to $A_\X$ --- and we say that $\X$ is \emph{twin-free} if no such pair of variables exists.
Note that it is easy to check in polynomial time whether $x$ and $y$ are twins ---
simply solve $A_\X$ with $(\hat{x}, \hat{y})$ set to $(0, 1)$ and $(1, 0)$.

Fix a twin-free instance $\X = (X, C)$ of monotone 2-in-3-SAT and consider the vector space $\F_2[X]$, i.e.~the space of formal linear combinations of elements in $X$ with coefficients in $\F_2$. Each equation $x + y + z \equiv 0 \bmod 2$ in $A_\X$ can be associated with an element $x + y + z$ of $\F_2[X]$. We let $\langle C \rangle$ denote the subspace generated by all of these equations.

\begin{lemma}
For every solution $(\hat{v})_{v \in X}$ of $A_\X$, and for any $x_1 + \cdots + x_k \in \langle C \rangle$, we have $\hat{x}_1 + \cdots + \hat{x}_k \equiv 0 \bmod 2$.
\end{lemma}
\begin{proof}
The term $x_1 + \cdots + x_k$ can be formed by summing together multiple equations from $A_\X$. Hence it must equal 0 in any solution to $A_\X$. 
\end{proof}

Thus, whenever $\X = (X, C)$ is twin-free, for any distinct $x, y \in X$ we have $x + y \not \in \langle C \rangle$, and hence $x$ and $y$ are different \emph{as elements of $\F_2[X] / \langle C \rangle$}. Observing that $\F_2[X] / \langle C \rangle$ is just some finite-dimensional vector space of the form $\F_2^d$, we have the following. (We note that a similar but more general construction appears in~\cite[Proposition 6.6]{Brakensiek25:stoc}.)

\begin{lemma}\label{lem:alpha}
    Whenever $\X= (X, C)$ is twin-free, we can compute in polynomial time an integer $d$ and an \emph{injective} map $\alpha : X \to \F_2^d$ so that the following holds.
    For any $x_1, \ldots, x_k \in X$ with $\alpha(x_1) + \cdots + \alpha(x_k) = 0$, we have $\hat{x}_1 + \cdots + \hat{x}_k \equiv 0 \bmod 2$ in any solution $(\hat{v})_{v \in X}$ to $A_\X$. Furthermore, for every equation $x + y + z \equiv 0 \bmod 2$ in $A_\X$, we have $\alpha(x) + \alpha(y) + \alpha(z) = 0$.
\end{lemma}
\begin{proof}
    Note that
        $X \to \F_2[X] \to \F_2[X] / \langle C \rangle \cong \F_2^d$.
    We output this composite as $\alpha$. For every input in $X$ it is straightforward to see what it should be mapped to in $\F_2[X]$ and further in $\F_2[X] / \langle C \rangle$. Moreover, from the previous discussion it follows that $\alpha$ satisfies the required conditions. Finally, the image in~$\F_2^d$ can be computed simply by finding a basis for the quotient space, which can be done in polynomial time.
\end{proof}

Fix a twin-free instance $\X = (X, C)$, and let $\alpha$ be given by \Cref{lem:alpha}. Consider two variables $x, y \in X$.
If there exists an even number of neighbours\footnote{We say two variables are neighbours if they belong to the same clause.} of $x$, say $z_1, \ldots, z_{2k}$, so that $\alpha(z_1) + \cdots + \alpha(z_{2k}) = \alpha(y)$, then write $x \succeq y$. This suggestive notation has the following justification.

\begin{lemma}\label{lem:cycle}
    Suppose $\X = (X, C)$ is twin-free and $x \succeq y$. Then, in any solution $(\hat{v})_{v \in X}$ to $\X$, we have $\hat{x} \geq \hat{y}$.
\end{lemma}
\begin{proof}
    If $\hat{x} = 1$ then the claim holds, so suppose $\hat{x} = 0$. Since $(\hat{v})_{v \in X}$ is a solution to $\X$ (which, recall, is a 2-in-3-SAT instance), it follows that for all neighbours $z$ of $x$ we have $\hat{z} = 1$. Now, by assumption we have that $\alpha(z_1) + \cdots + \alpha(z_{2k}) = \alpha(y)$ i.e.~$\alpha(z_1) + \cdots + \alpha(z_{2k}) +\alpha(y) = 0$. Hence 
    \[
    \hat{y} \equiv \hat{y} + 2k \equiv \hat{y} + \hat{z}_1 + \cdots + \hat{z}_{2k} \equiv 0 \mod 2.
    \]
    Thus $\hat{y} = 0$ and $\hat{x} \geq \hat{y}$ as required.
\end{proof}

We can check whether $x \succeq y$ in polynomial time, as we will now describe. Let $z_1, \ldots, z_t$ be the neighbours of $x$. For $j \in [d]$ let $\alpha^j$ be defined so that $\alpha(x) = (\alpha^1(x), \ldots, \alpha^d(x))$.  We check whether there exist $b_1, \ldots, b_t \in \F_2$ so that $\sum_{i = 1}^t b_i = 0$ and, for all $j \in [d]$, we have $\sum_{i = 1}^t b_i \alpha^j(z_i) = \alpha^j(y)$.
If they exist, let $Z$ be the set of these $z_i$'s for which $b_i=1$. Clearly, the first sum guarantees that $Z$ has an even number of elements, while the others give that $\sum_{z \in Z}\alpha(z) = \alpha(y)$.
Therefore, $x \succeq y$ if and only if $b_1,\ldots,b_t$ exist, and we can decide that by solving a system of linear equations.

Call a sequence of variables $x_1, \ldots, x_k$ a \emph{cycle} if $x_1 \succeq \cdots \succeq x_k \succeq x_1$. If $\X$ does not admit any cycles, call it \emph{cycle-free}. With these definitions in place, we can finally apply \Cref{thm:AddComb} in the following theorem.

\begin{theorem}\label{thm:glue}
    Suppose $\X = (X, C)$ is an $n$-variable, $m$-clause instance of monotone 2-in-3-SAT that is twin-free and cycle-free. Then $m = O(n^{2 - \epsilon})$, for the same $\epsilon$ as in \Cref{thm:AddComb}.
\end{theorem}
\begin{proof}
    Let $\alpha, d$ be given by \Cref{lem:alpha}.
    Consider the relation $\succeq$. As $\succeq$ is assumed to be acyclic, there exists a topological sort of $V$ with respect to $\succeq$. In other words, we order $X = \{ x_1, \ldots, x_n \}$ in such a way that for $j \leq i$ we have $x_i \not \succeq x_j$.
    With this in hand, we apply \Cref{thm:AddComb} to $V = \{ \alpha(x_1), \ldots, \alpha(x_n) \}$ and $N_i = \{ \alpha(y) \mid \text{$y$ is a neighbour of $x_i$} \}$. Let us first check that the properties of \Cref{thm:AddComb} are satisfied.
    \begin{enumerate}[label={(\roman*)}]
        \item Consider any element $\alpha(y) \in N_i$. There exists $z$ so that $(x_i, y, z)$ is a clause of $\X$. The properties of $\alpha$ guarantee that $\alpha(x_i) + \alpha(y) + \alpha(z) = 0$, hence $\alpha(y) + \alpha(x_i) = \alpha(z) \in N_i$, as $z$ is also a neighbour of $x_i$. 
        \item Consider any $j \leq i$. We have $x_i \not \succeq x_j$ by our choice of ordering of $x_1, \ldots, x_n$. Hence, there is no collection of an even number of neighbours $z_1, \ldots, z_{2k}$ of $x_i$ so that $\alpha(x_j) = \sum_{\ell= 1}^{2k} \alpha(z_\ell)$. The set of possible sums on the right ranges over $\langle N_i + N_i \rangle$, thus we obtain that $\alpha(x_j) \not \in \langle N_i + N_i \rangle$.
    \end{enumerate}
    
    Hence, we can apply \Cref{thm:AddComb}, and thus we conclude that $\sum_{i = 1}^n \# \{ \alpha(y) \mid \text{$y$ is a neighbour of $x_i$} \} = O(n^{2-\epsilon})$ for $\epsilon \approx 0.0028$. Since $\alpha$ is injective, this is the same as saying that the total number of pairs of variables $(x, y)$ that are neighbours is at most $O(n^{2 - \epsilon})$.

    Observe that, for any four distinct variables $x, y, z, t \in X$, it is impossible that there is a clause on variables $x, y$, and $z$ and another clause on variables $x, y$, and $t$, as then $z$ and $t$ would be twins. In other words, each pair of variables that are in a clause together are in \emph{exactly} one clause together. Thus, the number of clauses is at most the number of such pairs --- whence the conclusion.
\end{proof}

We are now ready to prove \Cref{thm:SparsAlg}.

\begin{proof}[Proof of \Cref{thm:SparsAlg}]
    Suppose we are given an instance $\X$ of monotone 2-in-3-SAT. 
    We will construct $\sim$ by repeatedly merging pairs of variables that are the same in all solutions to $\X$. The relation $\sim$ will then be the transitive, reflexive closure of all the merges. 
    
    Let us now describe what variables we merge. While there are any twins $x$ and $y$, merge them. If there are no twins, but there exists a cycle $x_1\succeq \ldots\succeq x_k\succeq x_1$, then merge all the variables $x_1, \ldots, x_k$. Since variables in a twin-pair have the same value in all solutions to $A_\X$, they must have the same value in all solutions to $\X$ --- and the latter is true for all variables in a cycle as well, due to \Cref{lem:cycle}. Detecting twins can be done in polynomial time; furthermore, computing $\preceq$ and then finding cycles in it can also be done in polynomial time. Suppose we started with an $n$-variable instance. We get, at the end, a twin-free, cycle-free instance on at most $n$ variables. By \Cref{thm:glue} this instance has $O(n^{2 - \epsilon})$ clauses, as desired.
\end{proof}

\section{Monotone strong sparsification implies non-monotone}\label{sec:reduction}

\begin{theorem}
    Suppose that there is a polynomial-time strong sparsification algorithm $\mathscr{A}$ for monotone 1-in-3-SAT with performance $f(n)$. Then there is a polynomial-time strong sparsification algorithm for \emph{non-monotone} 1-in-3-SAT with performance at most $8f(2n)$.
\end{theorem}
\begin{proof}
Suppose we are given an instance $\X$ of non-monotone 1-in-3-SAT with variables $X = \{ x_1, \ldots, x_n \}$ and clauses $C$. Add variables $y_1, \ldots, y_n$, and create an instance $\Y$ with variable set $Y = \{ x_1, y_1, \ldots, x_n, y_n \}$ and whose set $C'$ of clauses is the same as in $\X$, but with the literal $\lnot x_i$ replaced by the variable $y_i$ in every clause.
    By construction $\Y$ is an instance of monotone 1-in-3-SAT with $2n$ variables; also every solution to $\X$ extends to a solution to $\Y$, by setting $y_i = \lnot x_i$.
    
    Thus, we can apply $\mathscr{A}$ to $\Y$ to create an equivalence relation $\sim$ on $Y$ such that, in any solution to $\Y$, if $v \sim w$ then $v$ and $w$ are assigned the same value. 
    Furthermore, by our assumption on $\mathscr{A}$, we have that the size of $C'/ \mathord{\sim}$ is  $f(2n)$. 
    By construction, in every solution to $\X$ extended to $\Y$, different values are assigned to $x_i$ and $y_i$. 
    Define a graph $G=(Y,E)$, where $(v,w)\in E$ if there exists $i\in[n]$ such that $v \sim x_i$ and $w \sim y_i$.
    Note that if $(v,w) \in E$ then, in any solution to $\X$ extended to $\Y$, different values are assigned to $v$ and $w$. 

    If the graph $G$ is non-bipartite then $\X$ is unsatisfiable, and our strong sparsification algorithm can vacuously return any equivalence relation; hence, suppose $G$ is bipartite, and denote by $A$ and $B$ its bipartition classes (if $G$ is not connected, fix one choice for $A$ and $B$).
    We write $x_i \sim_G x_j$ if $x_i$ and $x_j$ belong to the same connected component of $G$, and are both in $A$ or both in $B$. 
    If $x_i \sim_G x_j$, then there exists a path $x_i - z_1 - \ldots - z_{2k-1} - x_j$ in $G$, which implies that $x_i$ and $x_j$ are assigned the same value in each solution to $\X$.
    Since the above procedure can return $\sim_G$ in polynomial time, it is a strong sparsification algorithm for 1-in-3-SAT.  
    It only remains to prove that this algorithm has the advertised performance; i.e.~that $\X / \mathord{\sim_G}$ does not have too many clauses.
    
    We claim that each element of $C' / \mathord{\sim}$ corresponds to at most 8 clauses from $C / \mathord{\sim_G}$. 
    We will prove this by constructing a mapping that associates, with each clause $c \in C / \mathord{\sim_G}$, a non-empty subset $A(c) \subseteq C' / \mathord{\sim}$, in such a way that every clause from $C' / \mathord{\sim}$ belongs to at most 8 different sets $A(c)$. 
    We illustrate the construction by example. Consider a clause $(P, \lnot Q, R) \in C / \mathord{\sim_G}$ (and recall that $P, Q, R$~are equivalence classes of $\sim_G$). Define
    \[
    A(P, \lnot Q, R)
    =
    \{ ( [x_i]_\sim, [y_j]_\sim, [x_k]_\sim) \mid x_i \in P, y_j \in Q, x_k \in R \}.
    \]
    
    Consider some clause $(S, T, U)$ of $\Y / \mathord{\sim}$ (recall again that $S, T, U$ are equivalence classes of $\sim$). For each of the 8 sign patterns in $\{ +, -\}^3$ there is at most one clause $c = \X / \mathord{\sim_G}$ for which $(S, T, U) \in A(c)$ which follows that sign pattern.
     For example, for the sign pattern $(+, -, +)$, consider any $x_i \in S, y_j \in T, x_k \in U$, and note that there is at most one clause of form $c = ([x_i]_{\sim_G}, \lnot [x_j]_{\sim_G}, [x_k]_{\sim_G})$ in $C / \mathord{\sim_G}$. This is because, for any other 
    $x_{i'} \in P, y_{j'} \in Q, x_{k'} \in R$, we have $x_i \sim x_{i'}, y_j \sim y_{j'}, x_k \sim x_{k'}$ and hence $x_i \sim_G x_{i'}, x_j \sim_G x_{j'}, x_k \sim_G x_{k'}$. Noting that only such a clause can have $(S, T, U) \in A(c)$ and follow the sign pattern $(+, -, +)$ completes the proof.
\end{proof}

\section{Generalisations and lower bounds}\label{sec:general}
As already pointed out in the introduction, the notion of strong sparsification
can be defined for different types of satisfiability problems, and in particular
all constraint satisfaction problems (CSPs).

Let $\Gamma$ be a set of relations on a fixed set $D$. An instance of
$\CSP(\Gamma)$ consists of a set X of variables and a set C of constraints: a
constraint is a tuple $(x_1,\ldots,x_r,R)$, where $x_1,\ldots,x_r \in X$ and
$R\in \Gamma$ is an $\ar(R)$-ary relation, i.e., $R\subseteq D^{\ar(R)}$.
The task is to decide whether there exists an assignment $c:X \to D$ such that
for every $(x_1,\ldots,x_r,R)\in C$ we have that $(c(x_1),\ldots,c(x_r)) \in R$.

\begin{definition}
    A \emph{strong sparsification algorithm for $\CSP(\Gamma)$} is an algorithm
    which, given an instance $\X=(X,C)$ of $\CSP(\Gamma)$, outputs an equivalence relation ${\sim}$ on $X$ such
    that if $x \sim y$ then $x$ and $y$ have the same value in all solutions to $\X$.
    The performance of the algorithm is given by the number of constraints in $C / {\sim}$ as a function of $n = |X|$.
\end{definition}
We call the equivalence relation $\sim$ outputted by a strong sparsification algorithm a \emph{strong sparsifier}.

In this section, we study the \emph{mere existence} of strong sparsifiers for classic computational problems that can be phrased as CSPs. In particular, we will give combinatorial lower bounds for strong sparsification for monotone 1-in-3-SAT: there exist instances with $n$ vertices and $\Omega(n^{1.725\cdots})$ clauses where every strong sparsifier is trivial, i.e.~merges no pair of vertices.

Note that for a given instance of CSP, an optimal strong sparsifier can be always found in exponential time by simply listing all possible solutions and merging variables that are the same in all solutions. Thus, the corollaries of this section can be read equivalently as: no exponential-time strong sparsification algorithm with given performance exists.

We start with the following definition.

\begin{definition}\label{def:stableInstance}
   Let $\Gamma$ be a set of relations on a fixed set $D$. An instance $\X = (X, C)$ of $\CSP(\Gamma)$ is \emph{stable} if for every two distinct variables $x,y \in X$ there exists a solution $c$ to $\X$ with $c(x) \neq c(y)$. 
The instance~$\X$ is \emph{maximal} if no constraint can be added to $\X$ without eliminating at least one possible solution, i.e. if $(x_1,\ldots,x_r,R) \not \in C$, for some $x_1,\ldots,x_r \in X$, for some $R \in \Gamma$ of arity $r$, then there exists a solution $c$ to $\X$ such that $(c(x_1),\ldots,c(x_r)) \notin R$.
\end{definition}

We observe that the performance of a strong sparsification algorithm for $\CSP(\Gamma)$ is precisely given by the properties of the stable, maximal instances: any family of stable instances of $\CSP(\Gamma)$ gives a lower bound for the performance, and the elements of the family can be assumed to be maximal a fortiori.

One can construct stable maximal instances using the following correspondence.
We note that the instances described in \Cref{thm:genericEquivalence} below coincide with subinstances of the \emph{long-code construction}, often used in the CSP literature~\cite{BGS:98}.

\begin{definition}
Let $\Gamma = \{R_1, \ldots, R_k\}$ be a set of relations on domain $D$, where $R_i$ is a relation of arity $\ar(R_i)$. 
For a family $V \subseteq D^N$ of vectors and $i \in [k]$ we denote by $\M_i(V) \subseteq D^{\ar(R_i \times N}$ the
maximal family of matrices such that for every $M \in \M_i$ every row of $M$ is a vector from $V$ and each column of $M$ is contained in $R_i$.
\end{definition}

The connection between stable maximal instances and families $\M_i(V)$ of matrices is as follows.

\begin{theorem}\label{thm:genericEquivalence}
   Let $\Gamma = \{R_1, \ldots, R_k\}$ be a set of relations on domain $D$, where $R_i$ is a relation of arity $\ar(R_i)$. 
Every stable maximal instance $\X$ of $\CSP(\Gamma)$ can be represented as a unique family $V \subseteq D^N$ for some $N \geq 1$.
Conversely, to every family of vectors $V \subseteq D^N$ we can assign a stable
  maximal instance of $\CSP(\Gamma)$ on $|V|$ variables and $|\M_i(V)|$
  constraints using the relation $R_i$.
\end{theorem}

\begin{proof}
We start with observing that every maximal instance gives rise to a different set of solutions.
Indeed, if two instances $\X_1,\X_2$ have the same set of solutions, then they must have the same set of variables (up to isomorphism), and thus any clause that belongs to $\X_1$, but not $\X_2$, contradicts the maximality of $\X_2$ and vice versa.
Therefore, if $c_1, \ldots, c_N : X \to D$ are the solutions to a stable, maximal instance $(X,C)$ of $\CSP(\Gamma)$, for every $x \in X$ we can set $v_x=(c_1(x),\ldots, c_N(x)) \in D^N$. Since the instance is stable, for each element of $X$ we get a different vector, which proves the first statement.

For the converse, we let the variables of the constructed instance be given by $V$. Now, for every $i \in [k]$, and every tuple $(v_1, \ldots, v_{\ar(R_i)})$ of elements of $V$, we add $(v_1, \ldots, v_{\ar(R_i)},R_i)$ to the set of clauses if and only if the vectors $v_1,\ldots, v_{\ar(R_i)}$ are precisely the consecutive rows of some matrix from $\M_i(V)$.

Since all possible constraints of this form are added, the created instance $(V,C)$ must be maximal. Observe that $(V,C)$ has at least $N$ solutions, $c_1, \ldots, c_N : V \to D$, where $c_i(v)$ is the $i$-th coordinate of~$v$. That these are solutions is certain by construction. The existence of these solutions ensures that no two variables can be merged --- since if $v, v' \in V$ and $v \neq v'$, then they must differ in some coordinate, say $i$, and hence they are assigned different values in $c_i$.
\end{proof}

Thus, by \Cref{thm:genericEquivalence}, estabilishing lower bounds for strong sparsification for different CSPs can be done by finding appropriate families of matrices. Observe that \Cref{cor:genericLowerBound} generalises \Cref{rem:lowerbound}, that yields the $n^{\log_2 3}$ lower bound for 1-in-3 SAT, to any $\Gamma$. 
Later, in \Cref{thm:final1inkbound} we will show an even stronger strong sparsification for 1-in-3 SAT lower bound.  

\begin{corollary}\label{cor:genericLowerBound}
  Let $\Gamma=\{R_1, \ldots, R_k\}$ be relations on domain $D$. Then,
  $\CSP(\Gamma)$ does not have a strong sparsification algorithm with performance $O(n^{\log_{|D|} \max_i |R_i|})$.
\end{corollary}
\begin{proof}
   As noticed before, any family of stable instances of $\CSP(\Gamma)$ gives a lower bound for the performance of any strong sparsification algorithm for $\CSP(\Gamma)$. Thus it is enough to the second part of \Cref{thm:genericEquivalence} to an arbitrary fixed $N \geq 1$ and $V = D^N$.
Since $|\M_i(D^N)|$ has $|R_i|^N=|D|^{N\cdot \log_{|D|}|R_i|}$ elements, the bound follows. 
\end{proof}

\begin{corollary}\label{cor:naeLowerBound}
    It is not possible to strongly sparsify ternary Not-All-Equal-3-SAT with
    performance better than $n^{\log_2 6} \approx n^{2.58}$. In general, it is
    not possible to strongly sparsify Not-All-Equal-$k$-SAT better than $n^{\log_2 (2^k - 2)}$, which is greater than $n^{k - 1}$ for $k \geq 3$.
\end{corollary}

\begin{corollary}\label{cor:coloringLowerBound}
    Graph $k$-colouring admits no strong sparsification with performance better than $n^{\log_k \binom{k}{2}} = n^{1 + \log_k \frac{k - 1}{2}}$.
\end{corollary}

We remark that the results in \Cref{cor:naeLowerBound} and \Cref{cor:coloringLowerBound} hold even if the sparsification algorithm is allowed to work in exponential time.
If we are interested in showing that it is not possible to find a strong sparsifier in polynomial time then even stronger bounds can be obtained. 
Indeed, Jansen and Pieterse show that no kernelisation of size $O(n^{k - 1})$ for Not-All-Equal-$k$-SAT~\cite{Jansen19:toct}, and of size $O(n^{2 - \epsilon})$ for $k$-colouring~\cite{JP19:algoritmica} is possible unless $\NP \subseteq \coNP / \poly$. 
Since strong sparsification is a restricted form of kernelisation, this applies to strong sparsification as well.

\begin{corollary}
$\LIN^k_2$, a system of homogeneous linear equations in $k$ variables modulo 2,
  cannot be strongly sparsified better than $n^{k-1}$.
\end{corollary}

If one traces through the exact instance we create in this corollary, it is a generalisation of the XOR-based instance from the introduction: the variable set is $V = \F_2^N$, and for every $x_1,\ldots, x_k \in \F_2^N$ with $x_1 + \cdots + x_k = 0$, we introduce the equation $x_1 + \cdots + x_k = 0$. Note that this is in severe contrast to other notions of sparsification (e.g.~computing a basis/kernelisation), where the number of constraints can be reduced further, to linear.

In what remains we will focus on improving the lower bound given by \Cref{cor:genericLowerBound} for $\ell$-in-$k$-SAT, where $1 \leq \ell < k$.

\begin{lemma}\label{lem:ellinkLowerBound}
    Let $1 \leq \ell < k$. Monotone $\ell$-in-$k$-SAT cannot be strongly sparsified with performance better than $n^{f(k, \ell)}$, where
    \[f(k,\ell)=\frac{{k \choose \ell}\ln {k \choose \ell}}{{k \choose \ell}\ln {k \choose \ell} - {k-1 \choose \ell}\ln {k-1 \choose \ell} - {k-1 \choose \ell-1} \ln {k-1 \choose \ell -1}}.\]
\end{lemma}
\begin{proof}
    Let $u = {k \choose \ell}$, and for every $t \geq 1$ define $\M_t$ be the family of matrices with $k$ rows and $ut$ columns that contain every tuple from $R_{\ell{\textrm{-in-}}k}$ as a column precisely $t$ times.
Define $V_t \subseteq \{0,1\}^{ut}$ to be the set of all possible rows of matrices from $\M_t$. 
By \Cref{thm:genericEquivalence}, every pair $(V_t,\M_t)$ gives rise to a stable maximal instance $\X_t$ of monotone $\ell$-in-$k$-SAT on $|V_t|$ variables and $|\M_t|$ constraints. 

There are $\binom{ut}{t, \ldots, t}$ matrices in $\M_t$, where 
\[
\binom{ut}{t, \ldots, t} = \frac{(ut)!}{(t!)^{k \choose \ell}}.
\]
Moreover, in each row, ${ut \choose t \cdot {k-1 \choose \ell-1}}$ different vectors can appear.
Now observe that, by applying Stirling's approximation, for $t \to \infty$ and natural numbers $u_1 + \cdots + u_k = u$ we have
\begin{multline*}
\ln \binom{ut}{t u_1, \ldots, t u_k} = ut \ln ut - ut - \left(\sum_{i = 1}^k t u_i \ln t u_i\right) + \left(\sum_{i = 1}^k t u_i\right) + O(\ln t) = \\
ut \ln t + ut \ln u - \left(\sum_{i = 1}^k tu_i \ln t + t u_i \ln u_i\right) + O(\ln t)
= t\ln\frac{u^u}{u_1^{u_1} \cdots u_k^{u_k}} + O(\ln t).
\end{multline*}

Thus, the number of constraints in $\X_t$ is at least
\begin{multline*}
        \log_{|V_t|} |\M_t| = \frac{\ln {ut \choose t, \ldots, t}}{\ln {ut \choose t \cdot {k-1 \choose \ell-1}}} \xrightarrow{t \to \infty} \frac{t \ln u^u}{t \ln \frac{u^u}{{k-1 \choose \ell}^{k-1 \choose \ell} \cdot {k-1 \choose \ell -1}^{k-1 \choose \ell-1}}} = \frac{{k \choose \ell}\ln {k \choose \ell}}{{k \choose \ell}\ln {k \choose \ell} - {k-1 \choose \ell}\ln {k-1 \choose \ell} - {k-1 \choose \ell-1} \ln {k-1 \choose \ell -1}},
    \end{multline*}
as desired.
\end{proof}

For the particular case $\ell = 1$, note that this becomes
\begin{equation}\label{eq:fEq}
f(k, 1)
=
\ln_{k^k / (k-1)^{(k-1)}} k^k.
\end{equation}

One can observe that all 1-in-$k$-SAT, Not-All-Equal-$k$-SAT and $\LIN_2^k$ have similar properties: they cannot be strongly sparsified with performance $n^{(1 - \epsilon)k}$ for $\epsilon > 0$ as $k \to \infty$.

Finally, we will show that the bound from~\Cref{lem:ellinkLowerBound} is optimal
in the case of $\ell=1$. In other words, \emph{every} instance of 1-in-$k$-SAT
\emph{can} (in principle) be sparsified with performance $n^{f(k, 1)}$.
%, cf.~\eqref{eq:fEq}. 
Indeed, this sparsification can be found in exponential time by brute force but
it is unclear how to achieve it in polynomial time. Note that for $k=3$ we get
performance $n^{\log_{27 / 4} 27}\approx n^{1.725\ldots}$, as mentiond in the
introduction.

We use the results of~\cite{Kane17:ejc,Ivanisvili17:ejc}. For this, recall that for functions $f, g : \{0, 1\}^m \to \mathbb{R}$, their convolution is defined by
\[
(f * g)(x_1, \ldots, x_m) = \sum_{\substack{y_1, \ldots, y_m, z_1, \ldots, z_m \in \{0, 1\} \\ y_i + z_i = x_i}} f(y_1, \ldots, y_m) g(z_1, \ldots, z_m),
\]
and the $p$-norm of $f$ is given by
\[
|| f ||_p = \left(\sum_{x_1, \ldots, x_m \in \{0, 1\}} f(x_1, \ldots, x_m)^p \right)^{1/p}.
\]
In particular, if $f, g$ are indicator functions for families of sets then $f * g$ counts the number of ways sets can be written as a disjoint union of one set from each family;  moreover, in this case $|| f ||_p$ is equal to the size of the represented set family, raised to the power $1/p$.
\begin{theorem}[\cite{Ivanisvili17:ejc}]
    For $m \geq 1$ and $f_1, \ldots, f_n : \{0, 1\}^m \to \mathbb{R}$, we have
    \[
    (f_1 * \cdots * f_n)(1,\ldots,1) \leq \prod_{i = 1}^n || f_i ||_{p_n},
    \]
    where $p_n = n / f(n, 1)$, and $f$ is taken from~\eqref{eq:fEq}.
\end{theorem}

By applying this result to the case where $f_1 = \cdots = f_n$ is an indicator function of a family of sets, we obtain the following.

\begin{corollary}
    If $\mathcal{F}$ is a family of subsets of $[m]$, then $[m]$ can be written as a disjoint union of $n$ sets from $\mathcal{F}$ in at most $|\mathcal{F}|^{n / p_n} = |\mathcal{F}|^{f(n, 1)}$ ways.
\end{corollary}

This is just a different way of stating that every $n$-vertex stable, maximal instance of 1-in-$k$-SAT has at most $n^{f(k, 1)}$ edges (indeed, simply replace every vector from the representation of \Cref{thm:genericEquivalence} with the set for which that vector is the indicator vector). For completeness, let us state this result.

\begin{theorem}\label{thm:final1inkbound}
    For any $k \geq 1$, any $n$-variable instance of monotone 1-in-$k$-SAT can be strongly sparsified with performance
    \[
      n^{\ln_{k^k / (k-1)^{(k-1)}} k^k}
    \]
    in exponential time and no better strong sparsification is possible.
\end{theorem}

{\small
\bibliographystyle{alphaurl}
\bibliography{bibliography}
}

\end{document}